\documentclass[envcountsect,envcountsame]{llncs}
\usepackage{amsmath,amssymb,bbm,graphicx}
\usepackage{enumerate}
\usepackage{epic,eepic}
\usepackage{tikz}

 \def\proofbox{\hbox{\vrule height 1.4ex depth 0.2ex width 0.4em}}\def\qed{}

\def\proofof#1{\par\noindent{\bfseries\itshape Proof of #1. }\def\qed{\hfill\proofbox\gdef\qed{}}}
\def\tclaim{\begin{description}\vspace{-2pt}\refstepcounter{theorem}\item[(\thetheorem)]}
\def\endtclaim{\vspace{-2pt}\end{description}}

\tikzstyle{every picture} = [>=latex]
\usepackage[author=,draft]{fixme}
\fxsetup{theme=color}

\usepackage{xspace}

\newcommand{\MSOi}{\ensuremath{\mathrm{MSO}_1}\xspace}
\newcommand{\MSOii}{\ensuremath{\mathrm{MSO}_2}\xspace}
\newcommand{\cardMSO}{\ensuremath{\mathrm{cardMSO}_1}\xspace}
\def\prebox#1{\mathop{\textsl{#1}}\nolimits}
\newcommand{\adj}{\prebox{adj}}
\newcommand{\true}{\prebox{tt}}
\newcommand{\false}{\prebox{ff}}

\renewcommand{\phi}{\varphi}

\newcommand{\Z}{\mathbb{Z}}

\newcommand{\cS}{{\cal S}}
\newcommand{\cT}{{\cal T}}
\newcommand{\cU}{{\cal U}}
\newcommand{\cV}{{\cal V}}

\newcommand{\cX}{{\cal X}}
\newcommand{\cY}{{\cal Y}}
\newcommand{\cZ}{{\cal Z}}

\begin{document}
\title{Expanding the expressive power of Monadic Second-Order logic on restricted graph classes}
\author{Robert Ganian\inst{1}
\and
Jan Obdr\v{z}\'alek\inst{2}}
\institute{
Vienna University of Technology, Austria\footnote{Robert Ganian acknowledges support by ERC (COMPLEX REASON, 239962).} \\
\email{rganian@gmail.com}
\and
Faculty of Informatics, 
Masaryk University, Brno, Czech Republic\footnote{Jan Obdr\v z\' alek 
is supported by the research centre Institute for 
Theoretical Computer Science (ITI), project No. P202/12/G061.}
\email{obdrzalek@fi.muni.cz}
}

\maketitle

\begin{abstract}
We combine integer linear programming and recent advances in Monadic Second-Order model checking to obtain two new algorithmic meta-theorems for graphs of bounded vertex-cover. The first shows that \cardMSO, an extension of the well-known Monadic Second-Order logic by the addition of cardinality constraints, can be solved in FPT time parameterized by vertex cover.
The second meta-theorem shows that the MSO partitioning problems introduced by Rao can also be solved in FPT time with the same parameter. 

The significance of our contribution stems from the fact that these
formalisms can describe problems which are W[1]-hard and even NP-hard on
graphs of bounded tree-width. Additionally, our algorithms have only an
elementary dependence on the parameter and formula. We also show that both
results are easily extended from vertex cover to neighborhood diversity.

\end{abstract}


\section{Introduction}

It is a well-known result of Courcelle, Makowski and Rotics that \MSOi (and
LinEMSO$_1$) model checking is in FPT on graphs of bounded clique-width
\cite{cmr00}. However, this leads to algorithms which are far from practical
-- the time complexity includes a tower of exponents, the height of which
depends on the \MSOi formula. Recently it has been shown that much faster
model checking algorithms are possible if we consider more powerful parameters such as vertex cover \cite{lam12} -- with only an elementary dependence of the runtime on both the \MSOi formula and parameter.

Vertex cover has been generally used to solve individual problems for which
traditional width parameters fail to help (see e.g. \cite{acs10,efg09,flm08,fgk10}).
Of course, none of these problems can be described by the standard \MSOi or LinEMSO$_1$ formalism.
This raises the following, crucial question: would it be possible to naturally extend the language of \MSOi to include additional well-studied problems without sacrificing the positive algorithmic results on graphs of bounded vertex-cover?

We answer this question by introducing \cardMSO (Definition \ref{def:card})
as the extension of \MSOi by linear cardinality constraints -- linear
inequalities on vertex set cardinalities and input-specified variables. The
addition of linear inequalities significantly increases the descriptive
power of the logic, and allows to capture interesting problems which are known to be hard on graphs of bounded tree-width. We refer to Section~\ref{sub:app} for a discussion of the expressive power and applications of \cardMSO, including a new result for the $c$-balanced partitioning problem (Theorem \ref{thm:cbal}).

The first contribution of the article lies in providing an FPT-time model checking algorithm for \cardMSO on graphs of bounded vertex cover. This extends the results on \MSOi model checking obtained by Lampis in \cite{lam12}, which introduce an elementary-time FPT \MSOi model checking algorithm parameterized by vertex cover. However, the approach used there cannot be straightforwardly applied to formulas with linear inequalities (cf. Section \ref{sec:cardmso} for further discussion).

\begin{theorem}
\label{thm:main}
There exists an algorithm which, given a graph $G$ with vertex cover of size
$k$ and a \cardMSO formula $\phi$ with $q$ variables, decides if $G\models \phi$ in time $2^{2^{O(k+q)}+|\phi|}+2^k|V(G)|$.
\end{theorem}

The core of our algorithm rests on a combination of recent advances in \MSOi model checking and the use of Integer Linear Programming (ILP). While using ILP to solve individual difficult graph problems is not new \cite{flm08}, the goal here was to obtain new graph-algorithmic meta-theorems for frameworks containing a wide range of difficult problems. The result also generalizes to the neighborhood diversity parameter introduced in \cite{lam12} and to \MSOii (as discussed in Section \ref{sec:con}).

In the second part of the article, we turn our attention to a different, already studied extension of \MSOi: the MSO partitioning framework of Rao \cite{rao07}.
MSO partitioning asks whether a graph may be partitioned into an arbitrary number of sets so that each set satisfies a fixed \MSOi formula, and has been shown to be solvable in XP time on graphs of bounded clique-width.
Although MSO partitioning is fundamentally different from \cardMSO and both formalisms expand the power of \MSOi in different directions, we show that a combination of \MSOi model checking and ILP may also be used to provide an efficient FPT model-checking algorithm for \MSOi partitioning parameterized by vertex-cover or neighborhood diversity. 

\begin{theorem}
\label{thm:msop}
There exists an algorithm which, given a graph $G$ with vertex cover of size
$k$ and a MSO partitioning instance $(\phi,r)$ with $q$ variables, decides if $G\models (\phi,r)$ in time $2^{2^{O(q2^k)}}\cdot |(\phi, r)|+2^k|V(G)|$.
\end{theorem}

\section{Preliminaries and Definitions}

\label{sec:pre}

\subsection{Vertex cover and types}

\noindent In the following text all graphs are simple and without loops. For a graph
$G$ we use $V(G)$
and $E(G)$ to denote the sets of its vertices and edges, and use $N(v)$ to
denote the set of neighbors of a vertex $v\in V(G)$.

The graph parameter we are primarily interested in is vertex cover.
A key notion related to graphs of bounded vertex cover is the notion of
a vertex type.

\begin{definition}[\cite{lam12}]
  Let $G$ be a graph. Two vertices $u,v\in V$ are of the same \emph{type}
  $T$ if $N(u)\setminus\{v\}=N(v)\setminus\{u\}$. We use $\cT_G$ to denote
  the set of all types of $G$ (or just $\cT$ if $G$ is clear from the
  context).
\end{definition}

Since each type is associated with its vertices, we also use $T$ to denote the set of
vertices of type $T$. Note that then $\cT_G$ forms a partition of the set
$V(G)$. 

For the sake of simplicity, we adopt the convention that, on a graph with a
fixed vertex cover $X$, we additionally separate each cover vertex into its
own type. Then it is easy to see that each type is an independent set, and a
graph with vertex cover of size  $k$ has at most $2^{k}+k$ types. 

It is often useful to divide vertices of the same type further into
subtypes. The subtypes are usually identified by a system of sets, and all
subtypes of a given type form a partition of that type:
\begin{definition} 
  Let $G$ be a graph and $\cU\subseteq
  2^{V(G)}$ a set of subsets of $V(G)$. Then two vertices $u,v\in V(G)$
  \emph{are of the same subtype} (w.r.t. $\cU$) if $u,v\in T$ for some
  $T\in\cT_G$ and $\forall U\in\cU.u\in U \iff v\in U$. We denote by
  $\cS_T^\cU$ the set of all subtypes of a type $T\in\cT_G$, and also define
  the set of all subtypes of (w.r.t. $\cU$) as $\cS_G^\cU$. (If $G$ and $\cU$
  are clear form the context, we may write $\cS$ instead of $\cS_G^\cU$)
\end{definition}

Finally, notice that $|\cS_G^\cU|\leq 2^{|\cU |}|\cT_G|$.

\subsection{\MSOi and its cardinality extensions}
\noindent Monadic Second Order logic (\MSOi) is a well established logic of graphs. It is the extension
of first order logic with quantification over vertices and sets of vertices.
\MSOi in its basic form can only be used to describe decision problems. To
solve optimization problems we may use LinEMSO$_1$ \cite{cmr00}, which is capable of finding maximum- and minimum-cardinality sets
satisfying a certain \MSOi formula. This is useful for providing simple
descriptions of well-known optimization problems such as Minimum Dominating
Set ($\adj$ is the adjacency relation):

$$\mbox{Min}(X)\!: \forall a \exists b\in X\!\!: (adj(a,b) \vee a=b)$$

The crucial point is that LinEMSO$_1$ only allows the optimization of set
cardinalities over all assignments satisfying a \MSOi formula. It is not
possible to use LinEMSO$_1$ to place restrictions on cardinalities of sets considered in the formula. 
In fact, such restrictions may be used to describe problems which are W[1]-hard on graphs of bounded tree-width, whereas all LinEMSO$_1$-definable problems may be solved in FPT time even on graphs of bounded clique-width \cite{cmr00}.

In this paper we define \cardMSO, an extension of \MSOi which allows
restrictions on set cardinalities. 

\begin{definition}[\cardMSO]
\label{def:card}
  The language of \cardMSO logic consists of expressions built from the following
  elements:
  \begin{itemize}
  \item variables $x,y\dots$ for vertices, and $X,Y\dots$ for sets of vertices
  \item the predicates $x\in X$ and $\adj(x,y)$ with the standard meaning
  \item equality for variables, quantifiers $\forall,\exists$ and the
    standard Boolean connectives
  \item $\true$ and $\false$ as the standard valuation constants
    representing true and false 
  \item the expressions $[\rho_1 \leq \rho_2]$, where the syntax of the $\rho$ expressions
    is defined as $\rho ::= n \mid |X| \mid \rho + \rho$ , where $n\in \Z$
    ranges over integer constants and $X$ over (vertex) set variables. 
  \end{itemize}

  We call expressions of the form $[\rho_1 \leq \rho_2]$ \emph{linear 
  (cardinality) constraints}, and write $[\rho_1 = \rho_2]$ as a shorthand for $[\rho_1\leq\rho_2]
  \land [\rho_2\leq\rho_1]$, and $[\rho_1 < \rho_2]$ for $[\rho_1\leq\rho_2]
  \land \neg[\rho_2 \leq \rho_1]$. A
  \emph{formula} $\phi$ of \cardMSO is an expression of the form $\phi =
  \exists Z_1\ldots \exists Z_m.\overline{\phi}$ such that $\overline{\phi}$
  is a \MSOi formula and $Z_1,\ldots, Z_m$ are the
  only variables which appear in the linear constraints.

  To give the \emph{semantics} of \cardMSO it is enough to define the
  semantics of cardinality constraints, the rest follows the standard \MSOi
  semantics. Let $\cV: \cX \to \Z$ be a valuation of set variables. Then the
  truth value of $[\rho_1 \leq \rho_2]$ is obtained be replacing each
  occurrence of $|X|$ with the cardinality of $\cV(X)$ and evaluating the
  expression as standard integer inequality.
\end{definition}

To give an example, the following \cardMSO formula is true if, and only if, a
graph is bipartite and both parts have the same cardinality:

\[
\exists X_1\exists X_2.(\forall v\in V. (v\in X_1 \iff \neg v\in X_2))\land
[|X_1|=|X_2|] \land \]
\[ (\forall u\in V. (\adj(u,v)\implies ((u\in
X_1 \land v\in X_2)\lor(u\in X_2 \land v\in X_1) ) )
\]

For a \cardMSO formula $\phi=\exists Z_1\ldots \exists Z_m.\overline{\phi}$ we call
$\exists Z_1\ldots \exists Z_m$ the \emph{prefix} of $\phi$, and the
variables $Z_i$ \emph{prefix variables}. We also put
$\cZ(\phi)=\{Z_1,\ldots,Z_m\}$, and often write just $\cZ$ if $\phi$ is
clear from the context.
Note that, since all prefix variables are existentially quantified set
variables, checking whether $G\models \phi$ (for some graph $G$) is
equivalent to finding a variable assignment $\chi: \cZ \to 2^{V(G)}$
such that $G\models_\chi\overline{\phi}$. We call such $\chi$ the
\emph{prefix assignment} (for $G$ and $\phi$). Note that the sets
$\chi(Z_i)$ can be used to determine subtypes, and therefore we often write
$\cS_G^\chi$ with the obvious meaning.

\subsection{ILP Programming}
\noindent Integer Linear Programming  (ILP) is a well-known framework for formulating problems, and will be used extensively in our approach. We provide only a brief overview of the framework:

\begin{definition}[p-Variable ILP Feasibility (p-ILP)]
Given matrices $A\in \mathbb{Z}^{m\times p}$ and $b\in \mathbb{Z}^{m\times
  1}$, the \emph{p-Variable ILP Feasibility (p-ILP) problem} is whether there exists a vector $x\in \mathbb{Z}^{p\times 1}$ such that  $A\cdot x\leq b$. The number of variables $p$ is the parameter.
\end{definition}

Lenstra \cite{len83} showed that p-ILP, together with its optimization
variant p-OPT-ILP, can be solved in FPT time. His running time was
subsequently improved by Kannan \cite{kan87} and Frank and Tardos
\cite{ft87}.

\begin{theorem}[\cite{len83,kan87,ft87,flm08}]
\label{thm:pilp}
p-ILP and p-OPT-ILP can be solved using $O(p^{2.5p+o(p)}\cdot L)$ arithmetic operations in space polynomial in $L$, $L$ being the number of bits in the input.
\end{theorem}

\section{\cardMSO Model Checking}
\label{sec:cardmso}

The main purpose of this section is to give a proof of
Theorem~\ref{thm:main}. The proof builds upon the following result of
Lampis:

\begin{lemma}[\cite{lam12}]
\label{lem:lam}
  Let $\phi$ be an $\MSOi$ formula with $q_S$ set variables and $q_v$
  vertex variables. Let $G_1$ be a graph, $v\in V(G_1)$ a vertex of type $T$ such
  that $|T|>2^{q_S}\cdot q_v$, and $G_2$ a graph obtained from $G_1$ by
  deleting $v$. Then $G_1\models \phi$ iff $G_2\models \phi$.   
\end{lemma}

In other words, a formula $\phi$ of \MSOi cannot distinguish between two
graphs $G_1$ and $G_2$ which differ only in the cardinalities of some types,
as long as the cardinalities in both graphs are at least $2^{q_S}\cdot q_v$.


\setcounter{footnote}{0}
This gives us an efficient algorithm for model checking $\MSOi$ on graphs of
bounded vertex cover: We first ``shrink'' the sizes of types to
$2^{q_S}\cdot q_v$ and then recursively evaluate the formula, at each
quantifier trying all possible choices for each set and vertex variable\footnote{Note that both Lemma \ref{lem:lam} and Theorem \ref{thm:lam} implicitly utilize the symmetry between vertices of the same type.}.

\begin{theorem}[\cite{lam12}]
\label{thm:lam}
  There exists an algorithm which, for a \MSOi sentence $\phi$ with $q$
  variables and a graph $G$ with $n$ vertices and vertex cover of size at most $k$,
  decides $G\models \phi$ in time $2^{2^{O(k+q)}}+O(2^kn)$.
\end{theorem}

However, a straightforward adaptation of the approach sketched above does
not work with linear constraints. To see this, simply consider e.g. the
formula $\exists Z_1\exists Z_2 .[|Z_1|=|Z_2|+1]$. Changing the cardinality
of $Z_1$ by even a single vertex can alter whether the linear constraint is
evaluated as true or false, even if $|Z_1\cap T|$ is large for some type
$T$. On the other hand, observe that the truth value of a linear inequality
$[\rho_1 \leq \rho_2]$ depends only on the prefix variables, not on the rest
of the formula. With this in mind, we continue by sketching the general strategy for
proving Theorem \ref{thm:main}:



Given a graph $G$ and a formula $\phi$ we begin by creating the graph
$G_\phi$ from $G$ by reducing the size of each type to $2^{q_S}\cdot
q_v$. Since this construction can impact the possible values of linear
constraints in $\phi$, we replace each linear constraint with either $\true$
or $\false$, effectively claiming which linear constraints we expect to be
satisfied in $G$ (for some assignment to prefix variables). We try all $2^l$
possible truth valuations of linear constraints.

For each \MSOi formula $\psi$ obtained from $\phi$ by fixing some truth valuation of linear
constraints we now check whether $G_\phi\models \psi$, generating all prefix
assignments $\chi$ for which $G_\phi\models_\chi \overline{\psi}$. The
remaining step is to check whether some prefix assignment (in $G_\phi$) can be
extended to a prefix assignment in $G$ in such a way that $\psi$ would still hold in
$G$ and all linear cardinality constraints would evaluate to their guessed values. This
check is performed by the construction of an p-ILP formulation which is feasible if,
and only if, there is such an extension.

We will now formalize the proof we have just sketched. First, we need a few
definitions. We start by formalizing the process of ``shrinking'' (some types
of) a graph.

\begin{definition}
  Given a graph $G$ and a $\cardMSO$ formula $\phi=\exists Z_1\dots \exists
  Z_m.\overline{\phi}$ with $q_v$ vertex and $m+q_S$ set variables, we
  define the \emph{reduced graph $G_\phi$} to be the graph obtained from $G$
  by the following prescription:
  \begin{enumerate}
  \item For each type $T\in\cT_G$ s.t. $|T|>2^{q_S+m}q_v$ we delete
    the ``extra'' vertices
    of type $T$ so that exactly $2^{q_S+m}q_v$ vertices of this type remain, and
  \item we take the subgraph induced by the remaining vertices.
  \end{enumerate}
\end{definition}

Note that vertices of a type with cardinality at most $2^{q_S+m}q_v$ are never deleted
in the process of ``shrinking'' $G$, and $|V(G_\phi)|\leq |\cT_{G_\phi}|\cdot 2^{q_S+m}q_v$. Next we formalize the process of fixing
the truth values of linear cardinality constraints.

\begin{definition}
  Let $l(\phi)=\{l_1,\ldots,l_k\}$ be the list of all linear cardinality
  constraints in the formula $\phi$. Let $\alpha: l(\phi) \to \{\true,
  \false\}$, called the \emph{pre-evaluation function}, be an assignment of truth
  values to all linear constraints. Then by $\alpha(\phi)$ we denote
  the formula obtained from $\phi$ by replacing each linear constraint
  $l_i$ by $\alpha(l_i)$, and call $\alpha(\phi)$ the\emph{ pre-evaluation of
  $\phi$}. Note that $\alpha(\phi)$ is a $\MSOi$ formula.
\end{definition}

As we mentioned earlier, the truth value for each linear cardinality
constraint depends only on the values of prefix variables. Therefore all
linear constraints can be evaluated once we have fixed a prefix
assignment. We say that a prefix assignment $\chi$, of a \cardMSO formula
$\phi$, \emph{complies with} a pre-evaluation $\alpha$ if each linear
constraint $l\in l(\phi)$ evaluates to true (under $\chi$) if, and only if,
$\alpha(l)=\true$.

We also need a notion of extending a prefix assignment for
$G_\phi$ to $G$. In the following definition we use the implicit matching
between the subtypes $S$ of $G$ and the subtypes $S_\phi$ of its subgraph $G_\phi$.
\begin{definition}
  Given a graph $G$ and a $\cardMSO$ formula $\phi=\exists Z_1\dots \exists Z_m.\overline{\phi}$ with $q_v$ vertex and $q_S$ set variables in $\overline{\phi}$, we say that a prefix assignment $\chi$ for $G$ extends a prefix assignments $\chi_\phi$ for $G_\phi$ if for all $S\in\cS^{\chi}_G$:
  \begin{enumerate}
  \item $S=S_\phi$ if $|S_\phi|\leq 2^{q_S}q_v$
  \item $S\supseteq S_\phi$ if $|S_\phi|> 2^{q_S}q_v$
  \end{enumerate}
\end{definition}

Finally we will need the following statement, which directly follows from the proof
of Lemma~\ref{lem:lam} \cite{lam12}: 

\begin{lemma}
\label{lem:lampis_qv}
Let $\phi=\exists Z_1\dots \exists Z_m.\overline{\phi}$ be an $\MSOi$ formula,  with $q_S$ set variables in $\overline{\phi}$ and $q_v$ vertex
variables, and let $\chi_1:\cZ\to 2^{V(G_1)}$ be a prefix assignment in
$G_1$. Let $v\in V(G_1)$ be a vertex of subtype $S\in \cS_{G_1}^\chi$ such that $|S|>2^{q_s}q_v$, and
$G_2$ a graph obtained from $G_1$ by deleting $v$. Then $G_1\models_{\chi_1}
\phi$ iff $G_2\models_{\chi_2} \phi$, where $\chi_2$ is the prefix
assignment induced by $\chi_1$ on $G_2$.
\end{lemma}

For the remainder of this section let us fix a $\cardMSO$ formula $\phi=\exists Z_1\dots \exists Z_m.\overline{\phi}$
with $q_v$ vertex variables, $q_S$ set variables in $\overline{\phi}$ and with linear cardinality
constraints $l(\phi)=\{l_1,\ldots,l_k\}$.
We are now ready to state the main lemma:

\begin{lemma}
\label{lem:ILP}
Let $G$ be a graph, $\phi$ be a \cardMSO formula, $\chi_\phi$ be a prefix
assignment for $G_\phi$, and $\alpha$ a pre-evaluation such that
$G_\phi\models_{\chi_\phi} \alpha(\overline{\phi})$. Then we can, in time $O(|\cT_G| \cdot 2^m|l(\phi)|)$,
construct a p-ILP formulation which is feasible iff $\chi_\phi$ can be
extended to a prefix assignment $\chi$ for $G$ such that (a) $\chi$ complies with
$\alpha$, and (b) $G\models_\chi \overline{\phi}$. Moreover, the formulation
has $|\cT_G| \cdot 2^{m}$ variables.
\end{lemma}
\begin{proof}
  We start by showing the construction of the p-ILP formulation. The set of
  variables is created as follows: For each subtype $S\in\cS_{G_\phi}^{\chi_\phi}$ we introduce a
  variable $x_S$ which will represent the cardinality of $S$ in $G$. There are three groups of constraints:

  1. We need to make sure that, for each type $T\in\cT_G$, the cardinalities of
  all subtypes of $T$ sum up to the cardinality of a type $T$. This is
  easily achieved by including a constraint $\sum_{S\subseteq T}x_S=|T|$ for each
  type $T$ (note that here $|T|$ is a constant).

  2. We need to guarantee that $\chi$ extends $\chi_\phi$. Therefore we include
  $x_S=|S_\phi|$ for each subtype with $|S_\phi|\leq 2^{q_S}q_v$, and
  $x_S>|S_\phi|$ if $|S_\phi|> 2^{q_S}q_v$.

  3. We need to check that $\chi$ complies with $\alpha$, i.e. that each linear
  constraint $l$ is either true or false based on the value of
  $\alpha(l)$. For each constraint $l$ we first replace each occurrence of
  $|Z_i|$ with the sum of cardinalities of all subtypes which are contained
  in $Z_i$, i.e. by $\sum_{S_\phi\subseteq Z_i}x_S$.  Then if $\alpha(l)=\true$,
  we simply insert the modified constraint into the formulation. Otherwise
  we first reverse the inequality (e.g. $>$ instead of $\leq$), and then
  also insert it.
  
  To prove the forward implication, let us assume that the p-ILP formulation
  is feasible. To define $\chi$ we start with $\chi=\chi_\phi$. Then for each
  subtype $S\in\cS_G$ if $x_S>|S_\phi|$ we add $x_S-|S_\phi|$ unassigned
  vertices of type $T$, where $T$ is the supertype of $S$. This is always
  possible thanks to constraints in 1. and 2. The constraints in 3. guarantee that
  $\chi$ complies with $\alpha$. Finally $G\models_\chi\overline{\phi}$ by
  Lemma~\ref{lem:lampis_qv}.

  For the reverse implication let $S\in\cS_G$ be the subtype identified by
  the set $\cY\subset \cZ$. Then we put $x_S=|\{v\in V(G)| \forall
  Z\in\cZ.v\in\chi(Z) \iff Z\in\cY \}$, and the p-ILP formulation is
  satisfiable by our construction. 
  Finally, it is easy to verify that the size of this p-ILP formulation is at most $O(|\cT_G| \cdot 2^{q_S}|l(\phi)|)$. \qed
\end{proof}

\bigskip

\proofof{Theorem \ref{thm:main}}
We start by constructing $G_\phi$ from $G$, which may be done by finding a
vertex cover in time $O(2^k\cdot n)$, dividing vertices into at most $2^k+k$ types (in linear time once we have a vertex cover) and keeping
at most $2^{q_S+m}q_v$ vertices in each type. 

Now for each pre-evaluation $\alpha:l(\phi)\to\{\true, \false\}$ we do the
following: We run the trivial recursive $\MSOi$ model checking algorithm on $G_\phi$,
by trying all possible assignments of vertices of $G_\phi$ to set and vertex
variables. Each time we find a satisfying assignment, we remember the values
of the prefix variables $\cZ$, and proceed to finding the next satisfying
assignment. Since the prefix variables of $\phi$ (and
$\alpha(\phi)$) are existentially quantified, their value is fixed before
$\alpha(\overline{\phi})$ starts being evaluated and therefore is the same at any
point of evaluating $\alpha(\overline{\phi})$. At the end of this stage we end up with
at most $(2^{|V(G_\phi)|})^m$ different
satisfying prefix assignments of $Z_1,\ldots,Z_m$ for each pre-evaluation $\alpha$.

We now need to check whether some combination of a pre-evaluation $\alpha$
and its satisfying prefix assignment $\chi_\phi$ from the previous step can be
extended to a satisfying assignment for $\overline{\phi}$ and $G$. This can be done by
Lemma~\ref{lem:ILP}.

To prove correctness, assume that there exists a satisfying assignment $\chi$ for $G$. We create $G'_{\phi}$ by, for each $T\in \cT_G$ such that $|T|>2^{q_S+m}q_v$, inductively deleting vertices from subtypes $S\subseteq T$ such that $|S|>2^{q_s}q_v$, until $|T|=2^{q_S+m}q_v$ for every $T$. Observe that $G'_{\phi}$ is isomorphic to $G_{\phi}$ and that there is a satisfying assignment $\chi'$ induced by $\chi$ on $G'_{\phi}$. Then applying the isomorphism to $\chi'$ creates a satisfying assignment $\chi_2$ on $G_{\phi}$, and Lemma \ref{lem:ILP} ensures that our p-ILP formulation is feasible for $\chi_2$.

To compute the time complexity of this algorithm, note that we first need
time $O(2^k\cdot n)$ to compute $G_\phi$. Then for each of the $2^{|l|}$
pre-evaluations we compute all the satisfying prefix assignments in time
$2^{2^{O(k+q_S+m)}q_v}$ by Theorem~\ref{thm:lam}. For each of the at most $(2^{|V(G_\phi)|})^m$ $=(2^{(2^k+k)\cdot 2^{q_S+m}q_v})^{m}$ satisfying prefix assignments for $G_\phi$, we check whether it can be extended to an assignment for $G$, which can be done in time at most $2^{2^{O(k+q_S+m)}}$ by applying Theorem~\ref{thm:pilp} on the p-ILP formulation constructed by Lemma \ref{lem:ILP}. We 
therefore need time $O(2^k\cdot n) + 2^m\cdot (2^{2^{O(k+q_S+m)}q_v+|l|} +
(2^{(2^k+k)\cdot 2^{q_S+m}q_v})^{m} \cdot 2^{2^{O(k+q_S+m)}})$, and the bound follows.  \qed

\medskip

\emph{\bfseries Remark:} The space complexity of the algorithm presented above may be improved by successively applying Lemma~\ref{lem:ILP} to each iteratively computed satisfying prefix assignment (for each pre-evaluation).

\medskip
Before moving on to the next section, we show how these results can be
extended towards well-structured dense graphs. It is easy to verify that the
only reference to an actual vertex cover of our graph is in Theorem
\ref{thm:lam} -- all other proofs rely purely on bounding the number of
types. In \cite{lam12} Lampis also considered a new parameter called
\emph{neighborhood diversity}, which is the number of different types of a
graph. I.e. graph $G$ has neighborhood diversity $k$ iff $|\cT_G|=k$. Since
there exist classes of graphs with unbounded vertex cover but bounded
neighborhood diversity (for instance the class of complete graphs),
parameterizing by neighborhood diversity may in some cases lead to better
results than using vertex cover.

\begin{corollary}
\label{cor:main}
There exists an algorithm which, given a graph $G$ with neighborhood diversity
$k$ and a \cardMSO formula $\phi$ with $q$ variables, decides if $G\models \phi$ in time $2^{k2^{O(q)}+|\phi|}+k\cdot poly(|V(G)|)$.
\end{corollary}

\begin{proof}
The proof is nearly identical to the proof of Theorem \ref{thm:main}. The only change is that we begin by computing the neighborhood diversity and the associated partition into types (which may be done in polynomial time, cf. Theorem 5 in \cite{lam12}), and we of course use the fact that the number of types is now at most $k$ instead of $2^k+k$.
\qed
\end{proof}

\section{Applications}
\label{sub:app}

\subsection{Equitable problems}

Perhaps the most natural class of problems which may be captured by \cardMSO
but not by $\MSOi$ (or even $\MSOii$) are equitable problems.  Equitable
problems generally ask for a partitioning of the graph into a (usually
fixed) number of specific sets of equal ($\pm 1$) cardinality.

\noindent\emph{Equitable c-coloring} \cite{mey73} is probably the most extensively studied
example of an equitable problem. It asks for a partitioning of a graph into
$c$ equitable independent sets and has applications in scheduling, garbage
collection, load balancing and other fields (see e.g. \cite{das06,baz01}).
While even equitable 3-coloring is W[1]-hard on graphs of bounded tree-width
\cite{fetal07}, equitable c-coloring may easily be expressed in \cardMSO:
$$\exists A,B,C: partition(A,B,C)\wedge \forall x,y: ((x,y\in A \vee x,y\in B \vee x,y\in C) \implies \neg adj(x,y)) $$
\vspace{-0.6cm}
$$\wedge equi(A,B)\wedge equi(A,C)\wedge equi(B,C)\mbox{, where}$$

\begin{itemize}
\item[$\bullet$] $partition(A,B,C)=\big(\forall x: (x\in A \vee \neg x\in B \vee \neg x\in C) \wedge (\neg x\in A \vee x\in B \vee \neg x\in C) \wedge (\neg x\in A \vee \neg x\in B \vee x\in C)\big)$.
\item[$\bullet$] $equi(T,U)=(\big[|T|=|U|+1\big] \vee \big[|T|+1=|U|\big] \vee \big[|T|=|U|\big])$.
\end{itemize}

\noindent\emph{Equitable connected c-partition} \cite{efg09} is another studied equitable
problem which is known to be W[1]-hard even on graphs of bounded path-width
but which admits a simple description in \cardMSO:
$$\exists A,B,C: partition(A,B,C)\wedge conn(A)\wedge conn(B)\wedge conn(C)$$
\vspace{-0.6cm}
$$\wedge equi(A,B)\wedge equi(A,C)\wedge equi(B,C)\mbox{, where}$$

\begin{itemize}
\item[$\bullet$] $conn(U)= \big( \forall T: (\forall x: x\in T \implies x\in U)\implies (T=U \vee (\neg \exists a: a\in T) \vee \exists a,b: a\in U\wedge \neg a\in T \wedge b\in T \wedge adj(a,b) \big)$.
\end{itemize}

\subsection{Solution size as input}
\cardMSO allows us to
restrict the set cardinalities by constants given as part of the input. For instance,
the formula below expresses the existence of an Independent Dominating Set of cardinality $k$:
$$\exists X: (\forall a,b\in X. \neg adj(a,b)) \land $$
\vspace{-0.6cm}
$$\land (\forall b\in V.b\in X \lor (\exists a\in X.\adj(a,b))) \land [|X|=k]$$

Notice that there is an equivalent \MSOi formula for any
fixed $k$. However, the number of variables in the \MSOi formula would depend
on $k$, which would negatively impact on the runtime of model checking. On
the other hand, using an input-specified variable only requires us to change
a constant in the p-ILP formulation, with no impact on runtime.

\subsection{c-balanced partitioning} 
Finally, we show an example of how our approach can be used to obtain new results even for optimization problems, which are (by definition) not
expressible by \cardMSO. While the presented algorithm does not rely directly on Theorem \ref{thm:main}, it is based on the same fundamental ideas.

The problem we focus on is $c$-balanced
partitioning, which asks for a partition of the graph into $c$ equitable
sets such that the number of edges between different sets is
minimized. The problem was first introduced in \cite{mac78}, has
applications in parallel computing, electronic circuit design and sparse
linear solvers and has been studied extensively (see e.g. \cite{ff12,ar06}).
The problem is notoriously hard to approximate, and while an exact XP
algorithm exists for the $c$-balanced partitioning of trees parameterized by
$c$ \cite{ff12,mac78}, no parameterized algorithm is known for graphs of
bounded tree-width. 

\begin{theorem}
\label{thm:cbal}
There exists an algorithm which, given a graph $G$ with vertex cover of size
$k$ and a constant $c$, solves $c$-balanced partitioning in time $2^{2^{O(k+c)}}+2^k|V(G)|$.
\end{theorem}

\begin{proof}
We begin by applying the machinery of Theorem \ref{thm:main} to the
  \cardMSO formula $\phi$ for equitable $c$-partitioning $\phi$:
$$\exists A,B,C: partition(A,B,C)\wedge equi(A,B)\wedge equi(A,C)\wedge equi(B,C)$$

Recall that this means trying all possible assignments of the
  $c$ set variables in $G_{\phi}$ and testing whether each assignment can be
  extended to $G$ in a manner satisfying $\phi$.  Unlike in Theorem
  \ref{thm:main} though, we need to tweak the p-ILP formulations to not
  only check the existence of an extension $\chi$ for our pre-evaluation
  $\alpha$, but also to find the $\chi$ which minimizes the size of the cut
  between vertex sets.

To do so, we add one variable $\beta$ into the formulation and use a p-OPT-ILP formulation minimizing $\beta$. We also add a single equality into the formulation to make $\beta$ equal to the size of the cut between the $c$ vertex sets. While it is not possible to count the edges directly, the fact that we always have a fixed satisfying prefix assignment in $G_\phi$ allows us to calculate $\beta$ as:
\begin{center}$\beta=const_0+\sum_{S\in U} const_Sx_S$, where \end{center} 

\begin{itemize}
\item $const_0$ is the number of edges between all pairs of cover vertices
  with different types (this is obtained from the prefix assignment in
  $G_\phi$),
\item $U$ is the set of subtypes which do not contain cover vertices (recall
  that each cover vertex has its own subtype),
\item $x_S$ is the ILP variable for the cardinality of subtype $S$
  (cf. Lemma \ref{lem:ILP}),
\item For each subtype $S$, $const_S$ is the number of adjacent vertices in
  the cover assigned to a different vertex set than $S$. The values of $const_S$
  depend only on the subtype $S$ and the chosen prefix assignment $\chi_\phi$ in $G_\phi$.
\end{itemize}

For each satisfying prefix assignment $\chi_\phi$ in $G_\phi$, the p-OPT-ILP formulation will not only check that this may be extended to an assignment $\chi$ in $G$, but also find the assignment in $G$ which minimizes $\beta$. 
All that is left is to store the best computed $\beta$ for each satisfying prefix assignment and find the satisfying prefix assignment with minimum $\beta$ after the algorithm from Theorem \ref{thm:main} finishes.

For correctness, assume that there exists a solution which is smaller than the minimal $\beta$ found by the algorithm. Such a solution would correspond to an assignment of $\phi$ in $G$, which may be reduced to a prefix assignment $\chi$ of a pre-evaluation $\alpha(\phi)$ in $G_\phi$. If we construct the p-ILP formulation for $\chi$ and $\alpha(\phi)$, then the obtained $\beta$ would equal the size of the cut. However, our algorithm computes the $\beta$ for all pre-evaluations and satisfying prefix assignments in $G_\phi$, so this gives a contradiction.
\qed
\end{proof}
\section{MSO Partitioning}
\label{sec:msop}
The MSO (or \MSOi) partitioning framework was introduced by Rao in \cite{rao07} and allows the description of many problems which cannot be formulated in MSO, such as Chromatic number, Domatic number, Partitioning into Cliques etc.
While a few of these problems (e.g. Chromatic number) may be solved on graphs of bounded tree-width in FPT time by using additional structural properties of tree-width, MSO partitioning problems in general are W[1]-hard on such graphs.

\begin{definition}[MSO partitioning]
Given a MSO formula $\phi$, a graph $G$ and an integer $r$, can $V(G)$ be partitioned into sets $X_1, X_2, \dots X_r$ such that $\forall i\in \{1,2,\dots, r\}: X_i \models \phi$~\emph{?}
\end{definition}

Similarly to Section \ref{sec:cardmso}, we will show that a combination of ILP and MSO model checking allows us to design efficient FPT algorithms for MSO partitioning problems on graphs of bounded vertex cover. However, here the total number of sets is specified on the input and so the number of subtypes is not fixed, which prevents us from capturing the cardinality of subtypes by ILP variables. Instead we use the notion of \emph{shape}:

\begin{definition}
Given a graph $G$ and a MSO$_1$ formula $\phi$ with $q_S, q_v$ set and vertex variables respectively, two sets $A,B\subseteq V(G)$ have the same \emph{shape} iff for each type $T$ it holds that either $|A\cap T|=|B\cap T|$ or both $|A\cap T|,|B\cap T|> 2^{q_S}q_v$.
\end{definition}

Let $A$ be any set of shape $s$. We define $|s\cap T|$, for any type $T$,
as:

\[ |s\cap T| = 
\begin{cases}
  |A\cap T| & \text{ if } |A\cap T|\leq 2^{q_S}q_v \\
  \top & \text{ otherwise} 
\end{cases}
\]


\noindent
Since $\phi$ is a \MSOi formula, from
Lemma \ref{lem:lam} we immediately get:

\begin{tclaim}
\label{c:p1}
For any two sets $A,B\subseteq V(G)$ of the same shape, it holds that $A\models \phi$ iff $B\models \phi$, and
\end{tclaim}
\begin{tclaim}
\label{c:p2}
Given a MSO formula with $q$ variables, a graph $G$ with vertex cover of
size $k$ has at most 
$(2^{q_S}q_v)^{2^k+k}$ distinct shapes.
\end{tclaim}

With these in hand, we may proceed to:

\vspace{0.2cm}
\proofof{Theorem \ref{thm:msop}}
First, we consider all at most $(2^{q_S}q_v)^{2^k+k}$ shapes of a set $X$.
For each such shape $s$, we decide whether a set $X_s$ of shape $s$ satisfies $\phi$ by Theorem \ref{thm:lam}. 
We then create an ILP formulation with one variable $x_s$ for each shape $s$ satisfying $\phi$. The purpose of $x_s$ is to capture the number of sets $X_s$ of shape $s$ in the partitioning of~$G$.

Two conditions need to hold for the number of sets of various shapes. First, the total number of sets needs to be $r$.
This is trivial to model in our formulation by simply adding the constraint that the sum of all $x_s$ equals $r$.

Second, it must be possible to map each vertex in $G$ to one and only one set $X$ (to ensure that the sets form a partition). Notice that if a partition were to only contain shapes with at most $2^{q_S}q_v$ vertices in each $T$, then the cardinality of $s\cap T$ would be fixed and so the following set of constraints for each $T\in \cT$ would suffice:

$\sum_{\forall s} x_s\cdot |s\cap T| = |T|$

However, in general the partition will also contain shapes with more than $2^{q_S}q_v$ vertices in $T$, and in this case we do not have access to the exact cardinality of their intersection with $T$. To this end, for each $T\in \cT$ we add the following two sets of constraints:
\begin{itemize}
\item[a)] $\sum_{\forall s:|s\cap T|\leq {2^{q_S}q_v}} x_s\cdot |s\cap T|+\sum_{\forall s:|s\cap T|=\top} x_s\cdot ({2^{q_S}q_v})\leq |T|$
\item[b)] $\sum_{\forall s:|s\cap T|\leq {2^{q_S}q_v}} x_s\cdot |s\cap T|+\sum_{\forall s:|s\cap T|=\top} x_s\cdot |T|\geq |T|$
\end{itemize}
Here a) ensures that a partitioning of $G$ into $\sum_{\forall s}x_s$ sets of shape $s$ can ``fit'' into each $T$ and b) ensures that there are no vertices which cannot be mapped to any set. Notice that if the partition contains any shape $s$ which intersects with $T$ in over $2^{q_S}q_v$ vertices then b) is automatically satisfied, since all unmapped vertices in $T$ can always be added to $s$ without changing $X_s\models \phi$.

If the p-ILP formulation specified above has a feasible solution, then we can construct a solution to $(\phi,r)$ on $G$ by partitioning $G$ as follows: For each 
shape $s$ we create sets $X_{s,1}\dots X_{s,x_s}$. Then in each type $T$ in $G$, we map $|T\cap s|$ yet-unmapped vertices to each set $X_{s,i}$. Constraints a) make sure this is possible. If there are any vertices left unmapped in $T$, then due to constraint b) there must exist some set $X'$ such that $|X'\cap T|>2^{q_S}q_v$. We map the remaining unmapped vertices in $T$ to any such set $X'$, resulting in a partition of $G$. Finally, the fact that each of our sets satisfies $\phi$ follows from our selection of shapes.

On the other hand, if a solution to $(\phi,r)$ on $G$ exists, then surely each set in the partition has some shape and so it would be found by the p-ILP formulation. The total runtime is the sum of finding the vertex cover, the time of model-checking all the shapes and the runtime of p-ILP, i.e. $O(2^k|V(G)) +2^{2^{O(k+q)}}\cdot q^{(2^k+k)}+q^{(2^k+k)\cdot q^{O(2^k+k)}}$.
\qed

Theorem \ref{thm:msop} straightforwardly extends to neighborhood diversity as well. Directly bounding the number of types by $k$ results in a bound of $(2^{q_S}q_v)^k$ on the number of distinct shapes in Claim \ref{c:p2}, and so we get:

\begin{corollary}
There exists an algorithm which, given a graph $G$ with neighborhood diversity at most $k$ and a MSO partitioning instance $(\phi,r)$ with $q$ variables, decides if $G\models (\phi,r)$ in time $2^{2^{O(qk)}}\cdot |(\phi, r)|+k|V(G)|$.
\end{corollary}

\section{Concluding Notes}
\label{sec:con}
The article provides two new meta-theorems for graphs of bounded vertex cover.
Both considered formalisms can describe problems which are W[1]-hard on graphs of bounded clique-width and even tree-width. 
On the other hand, we provide FPT algorithms for both \cardMSO and MSO partitioning which have an elementary dependence on both the formula and parameter (as opposed to the results of Courcelle et al. for tree-width).

The obtained time complexities are actually fairly close to the lower bounds provided in \cite{lam12} for \MSOi model checking (already $2^{2^{o(k+q)}}\cdot poly(n)$ would violate ETH); this is surprising since the considered formalisms are significantly more powerful than \MSOi. Our methods may also be of independent interest, as they show how to use p-ILP as a powerful tool for solving general model checking problems.

Let us conclude with future work and possible extensions of our results.
As correctly observed by Lampis in \cite{lam12}, any \MSOii formula can be expressed by \MSOi on graphs of bounded vertex cover. This means that an (appropriately defined)
{\ensuremath{\mathrm{cardMSO}_2}\xspace} or MSO$_2$ partitioning formula could be translated to an equivalent \cardMSO or MSO partitioning formula on graphs of bounded vertex cover. However, the details of these formalisms would
need to be laid out in future work.

Another direction would be to extend the results of Theorems \ref{thm:main} and \ref{thm:msop} to more general parameters, such as twin-cover \cite{gan11} or shrub-depth \cite{ganianetal12}. Finally, it would be interesting to extend \cardMSO to capture more hard problems. Theorem \ref{thm:cbal} provides a good indication that the formalism could be adapted to also describe a number of optimization problems on graphs.

\bibliographystyle{abbrv}

\bibliography{gtbibcardMSO}

\end{document}